\documentclass[10pt,journal]{IEEEtran}

\usepackage{balance}

\usepackage{amsmath,amsthm,amssymb}
\usepackage{subcaption}

\usepackage{tikz}
\usetikzlibrary{shapes.symbols,arrows}
\usetikzlibrary{matrix}
\usetikzlibrary{circuits.logic.US}
\usetikzlibrary{decorations.text}

\ifCLASSOPTIONcompsoc
  \usepackage[nocompress]{cite}
\else
  \usepackage{cite}
\fi
\usepackage{enumerate}
\usepackage{algorithm,algorithmic}
\usepackage{multirow}
\usepackage{epstopdf}
\graphicspath{{./figs/}}

\newtheorem{thm}{Theorem}
\newtheorem{lem}[thm]{Lemma}
\newtheorem{cor}[thm]{Corollary}
\newtheorem{definition}[thm]{Definition}

\newcommand{\IR}{\mathbb{R}}

\newcommand{\IZ}{\mathbb{Z}}

\usepackage{mathtools}

\newcommand{\Tosc}{T_{\text{osc}}}
\newcommand{\Tctr}{T_{\text{ctr}}}
\DeclareMathOperator{\md}{md}
\newcommand{\mdrcv}{\md_{\text{rcv}}}
\newcommand{\mdsnd}{\md_{\text{snd}}}
\newcommand{\tauwrt}{\tau_{\text{s}}}
\newcommand{\taurd}{\tau_{\text{r}}}
\newcommand{\Oscrcv}{\textsc{Osc}_{\text{rcv}}}
\newcommand{\Oscsnd}{\textsc{Osc}_{\text{snd}}}
\newcommand{\clkrcv}{\text{clk}_{\text{rcv}}}
\newcommand{\clksnd}{\text{clk}_{\text{snd}}}

\newcommand{\clocked}{\textsf{ClockedTh}}

\newcommand{\cont}{\textsc{ContTh}}

\newcommand{\metas}{\text{\texttt{M}}}

\newcommand{\thresh}{\mathcal{T}}
\newcommand{\discclk}{C}
\newcommand{\realclk}{c}
\newcommand{\freqclk}{\dot{c}}

\newcommand{\rapointer}{P_r}
\newcommand{\sapointer}{P_s}
\newcommand{\rap}{p_r}
\newcommand{\drap}{\dot{p}_r}
\newcommand{\sap}{p_s}

\newcommand{\ctrl}{\textsc{Ctrl}}
\newcommand{\rcv}{\textsc{Rcv}}
\newcommand{\snd}{\textsc{Snd}}
\newcommand{\buff}{\textsc{Buff}}

\newcommand{\flevel}{\textit{fill}}

\newcommand{\TIME}{\ENSURE{}}

\newcommand{\figref}[1]{Fig.~\ref{#1}}

\usepackage{color}
\usepackage{soul}

\newcommand{\drawsector}[6][]{
  \draw[#1] (#4:{#2-.5*#3}) arc [start angle = #4, delta angle=-#5, radius={#2-.5*#3}] --
      ++({#4-#5}:#3) arc [start angle = {#4- #5}, delta angle=#5, radius={#2+.5*#3}] --cycle;
  \draw[decorate,decoration={raise=-3pt, text along path, text=#6, text align={align=center}}]
      (#4:#2) arc(#4:(#4-#5):#2);
}

\begin{document}

\title{Synchronizer-free Digital Link Controller}

\author{Johannes Bund, Matthias~F\"ugger, Christoph Lenzen, Moti Medina%
  \thanks{Johannes Bund and Christoph Lenzen are with the Max Planck Institute for Informatics,
      Saarland Informatics Campus, Germany. Their work on this project has been supported by funding from the European Research Council (ERC) under the European Union's Horizon 2020 research and
innovation programme (grant agreement 716562). Email:~\texttt{\{jbund,clenzen\}@mpi-inf.mpg.de}}%
    \thanks{Johannes Bund is with the Saarbr\"ucken Graduate School of Computer Science, Germany.}%
    \thanks{Matthias~F\"ugger is affiliated with the CNRS \& LSV, ENS Paris-Saclay, Universit\'e Paris-Saclay \& Inria, France. He receives funding from DigiCosme and DEPEC MODE. Email:~\texttt{mfuegger@lsv.fr}}%
    \thanks{Moti Medina is with the Ben Gurion University of the Negev, Israel. He was partially supported  by the Israel Science Foundation grant No. 867/19. Email:~\texttt{medinamo@bgu.ac.il}}%
  \thanks{Manuscript received Mmmm DD, YYYY; revised Mmmm DD, YYYY.}%
}

\markboth{IEEE Transactions on Circuits and Systems~I,~Vol.~XX, No.~X, Mmmmm~YYYY}{Bund \MakeLowercase{\textit{et al.}}: Synchronizer-free Digital Link Controller}

\maketitle

\begin{abstract}
This work presents a producer-consumer link between two independent clock
domains. The link allows for metastability-free, low-latency, high-throughput
communication by slight adjustments to the clock frequencies of the producer
and consumer domains steered by a controller circuit.

Any such controller cannot deterministically avoid, detect, nor resolve
metastability. Typically, this is addressed by
synchronizers, incurring a larger dead time in the control loop.
We follow the approach of Friedrichs et al.~(TC 2018) who proposed
metastability-containing circuits. The  result is a simple control circuit that may
become metastable, yet deterministically avoids buffer underrun or overflow.
More specifically, the controller output may become metastable, but this may
only affect oscillator speeds within specific bounds.
In contrast, communication is guaranteed to remain metastability-free.

We formally prove correctness of the producer-consumer link and a possible
implementation that has only small overhead. With SPICE simulations of the
proposed implementation we further substantiate our claims. The simulation
uses 65\,nm process running at roughly 2\,GHz.
\end{abstract}


% BEGIN \input{intro.tex}
\section{Introduction}\label{sec:introduction}

Links that enable communication between different clock domains are an important ingredient in every Globally Synchronous Locally Asynchronous (GALS) system~\cite{teehan2007survey}.
This communication is performed in a ``producer-consumer'' manner:
  in one clock domain the producer pushes messages to the link, while in the other clock domain the consumer
  pulls messages from the other side.
Inherently, link implementations are susceptible to failures induced by metastable upsets;
even if such errors can be handled, they negatively impact the performance of the link.

Previous digital controller designs resort to different methods to deal with metastability:
  clock-masking~\cite{coates2003congestion}, clock-pausing \cite{mullins2007demystifying,najvirt2015synchronize},
  or adding synchronizers (while sacrificing latency) to maintain a realistic (yet finite)
  mean time between failures (MTBF) of the link~\cite{chelcea2004robust,coates2003congestion,dally2010even,dobkin2006high,jackson2016gradual}.
Downsides of these approaches are that synchronized fill level flags are inherently ``stale'' by the time they affect the
  system. This requires almost-full flags~\cite{coates2003congestion},
  long handshake latencies that increase the dead time and affect the latency and
  throughput, additional slack in a controller cycle accounting for metastability resolution
  time in controller's flip-flops or mutual exclusion (MUTEX) elements~\cite{dobkin2006high,jackson2016gradual}.

At the heart of the problems faced in these controllers lies the impossibility to solve discrete decision problems,
  e.g., writing to a cell at a certain clock tick or skipping a clock cycle, under
  continuous inputs (i.e., arbitrary phase shifts between producer and consumer clocks)
  within bounded time~\cite{Mar81}.
One way out of this impossibility is to resort to {\em end-to-end analog designs}, e.g., by letting
  an analog controller apply continuous phase shifts by (slightly) tuning the producer and/or
  consumer oscillator.
This comes at the burden of a fully-fledged analog design.
\begin{table}
  \caption{Performance and hardware overhead (buffer size $N$, gates, flip-flops, oscillator type) of the proposed
    controller with a tunable 2.0 to 2.3\,GHz oscillator, \cite{dally2010even}, and \cite{polzer2009metastability}.}
  \vspace{-0.4cm}
\label{table:circbounds}
\begin{center}%\scriptsize
\begin{tabular}{| c | c | c || c | c |}
\cline{3-5}
    \multicolumn{2}{c|}{} & this work & \cite{polzer2009metastability} & \cite{dally2010even}
\tabularnewline
\hline
 \multirow{3}{*}{Performance}
                                    & Latency [ns]  & $1$& $375$ & $1.3$
 \tabularnewline
                                    & Th.put [$\frac{\text{pkt}}{\text{ns}}$]  & $2$ & $\frac{1}{41}$ & $\frac{1}{1.3}$
 \tabularnewline
                                    & MTBF  & $\infty$ & $\infty$ & Finite
 \tabularnewline
\hline
 \multirow{4}{*}{Overhead}& $N$  & $2$ & $9$ & $2$
 \tabularnewline
 & \# Gates   & $8$ & $>100$ & $>100$
 \tabularnewline
                        & \# FFs  & $4$ & $>50$ & $>100$
 \tabularnewline
                                    & Osc.\,Type & tune & distr. & quartz \tabularnewline
\hline
\end{tabular}
\end{center}
\vspace{-0.4cm}
\end{table}

An interesting alternative was proposed in~\cite{sokolov2014towards}, where the authors advocate the
  use of asynchronous controllers, sensing and controlling analog processes.
With this approach, analog components are required at the controller interfaces only, and the controller
  itself is implemented by a digital asynchronous circuit.
For certain classes of controllers, this approach allows to completely circumvent metastable upsets within the
  controller circuit, essentially by allowing for the occurrence of (digital)
  controller outputs within a continuous time range, rather than at discrete
  clock ticks only.

%###
\subsubsection*{Contribution}
%###

We propose a fundamentally different approach, exemplifying it at the hand of
highly efficient link controllers:
  like~\cite{sokolov2014towards}, we replace large parts of a (conceptually) analog controller
  by standard digital circuitry.
However, we do \emph{not} resort to asynchronous circuits.
Instead, we allow unstable/metastable signal values within our circuit
  and treat them as a third ``logical'' value.
Clearly, care must be taken that such values do not ``infect'' the whole controller logic,
  leading to unconstrained control outputs.
For this purpose, we follow \cite{friedrichs18}, using the same worst-case
  propagation model and analysis to provably contain metastability.

Specifically, we propose a digital controller
  that drives tunable ring oscillators as presented in~\cite{ghai2009design} at the
  sender and receiver side and prove its correctness.
The controller is small in size, has low control latency and allows for small link buffers.
We show that this guarantees high throughput and low latency communication.
Most notably, while the controller may become metastable, we ensure that metastability is contained within
  the controller, and does not lead to metastable upsets, corruption, or drops of communicated data words
  in the ring buffer between the sender and receiver.
We complement our provable system guarantees with simulations (see Sec.~\ref{sec:sims}).

%###
\subsubsection*{Related Work and Comparison}\label{sec:related}
%###

There is a large body of work on links between clock domains, motivated by their central importance in GALS designs.
% We confine ourselves to presenting the closest related literature.
According to~\cite{teehan2007survey}, GALS systems can be classified by their clocking schemes: (i) pausible clocked systems, (ii) asynchronous systems with
  uncorrelated clocks, and (iii) loosely synchronous systems, with (partially) synchronized clocks.
We shortly review sender-receiver communication in these three approaches.

(i) \emph{Pausible clocking} overcomes synchronization issues by halting the clock until metastability is
  resolved~\cite{mullins2007demystifying}; e.g., the design in \cite{najvirt2015synchronize}
  guarantees no glitches on stopping and starting.
Metastability inside the control loop may lead to an arbitrary delay of the
final pulse on stopping. This requires that the clock cannot be started
  again before metastability has been resolved.

(ii) \emph{Uncorrelated clocks:} Communication between uncorrelated clock frequencies and phases is traditionally done
  by combining classical two-flop synchronizers with buffers and flow-control circuitry.
A downside of these approaches is that the latency and the throughput
are determined by the handshake cycle that has to include (at least)
two synchronizer cycles at both sides.
Clearly, also this approach has a non-zero upset probability and thus finite MTBF.
In~\cite{chelcea2004robust}, a mixed-clock first-in first-out pipeline (FIFO) with flow control logic is proposed.
Instead of classical handshaking, synchronized full/empty and almost full/empty signals are used.
The throughput is one data item per clock cycle until the almost full signal is raised;
  afterwards, the ``true'' full signal has to be considered, at the cost of increased latency and lower throughput.
The approach has finite MTBF.
In~\cite{coates2003congestion}, a ripple FIFO solution with almost full/empty signals is proposed.
The approach requires slow sender/receiver speeds compared to data propagation within the ripple FIFO.
Moreover, full/empty flags have to be synchronized, which leads to increased latency and finite MTBF.
In~\cite{dobkin2006high}, a locally delayed latching (LDL) approach is proposed: conflicting read/write operations are delayed by
  an asynchronous controller with a MUTEX element.
Controller latency is in the order of $20$ gate delays, and the minimum feasible clock cycle is no less than $69$ gate delays, accounting for
  enough time for the MUTEX to stabilize with high probability.
Gradual synchronization~\cite{jackson2016gradual} allows fine-grained interweaving of synchronization and computation, also shifting
  conflicting ripple FIFO requests by MUTEX elements at each stage.
Like synchronizer chains, this approach has finite MTBF that can be increased at the cost of higher latency.
Dally and Tell~\cite{dally2010even} propose a scheme in which the MTBF can be made arbitrarily large
  without increasing latency.
They use synchronizers to continually determine phase offsets between sender and receiver clocks only.
A drawback is that the frequency and phase measurement circuits require accurate phase tracking (64bit in their implementation)
  and can account for slow phase drifts only.

(iii) \emph{Loosely synchronous systems:} in contrast to (i) and (ii),
synchronizing clocks allows obtaining worst-case guarantees on latency and
throughput together with provable absence of metastable upsets. Our approach
also falls into this class. The closest work to our approach presumably is
proposed in \cite{polzer2009metastability}.
By using the distributed DARTS clock generation mechanism~\cite{FS12:DC}, a
buffer size of $9$ and latency of $9$ clock cycles was achieved for a
receiver-sender clock shift of $4$ ticks at around $25$\,MHz in an FPGA.
While these numbers clearly can be improved in ASIC designs, DARTS inherently is
slower than our approach.

Table~\ref{table:circbounds} shows a comparison of our controller with the
  most closely related works, \cite{polzer2009metastability} and~\cite{dally2010even} (cf.\ Section~\ref{sec:sims} for details).

Our link controller has some similarities to a phase locked loop (PLL) with an all-digital phase detector;
see e.g.\ \cite{CC03:pll,abadian2010new} for all-digital PLL designs.
We briefly summarize commonalities and differences in the
following.\footnote{We remark that the exposition does not rely on the
information given in this comparison, and it might be easier to follow the
comparison based on a more detailed understanding of our approach. Accordingly,
readers should feel free to skip to the next section and return later if
needed.}

Classical PPLs lock a slave clock to a typically more stable master clock.
In our case we do not distinguish between a slave and a master, but our controller treats both receiver and sender clocks equally;
  one might think of this as a ``peer-to-peer PLL''.
The reason is that our goal is not to stabilize the absolute frequency of a poor
clock by ensuring a bounded phase offset to a more stable master clock,
but rather to bound the phase offset between a sender and receiver clock of
similar quality. Additionally we provide lower and upper bounds on
the frequency of the clocks which are close to the frequency bounds free-running
oscillators of the same quality have.
For example, this is useful when communicating with the environment.

The initial stage of a classical PLL is a phase frequency detector (PFD), which measures the phase difference between the
  master and slave clock signals.
Designs range from conventional PFDs, which measure negative and positive phase offsets on separate binary output signals
  by producing pulses whose width is the negative/positive phase offset, to more advanced setups~\cite{pan2007fast}.
Phase differences are then either forwarded to charge pumps (analog PLLs) \cite{bashir2009fast} or converted to digital counter offsets (digital PLLs).
For the latter, an unstable phase difference poses a risk for increased power consumption and likelihood of metastable upsets; see \cite{abadian2010new},
  where a filter on phase difference signals for a low-power digital PLL is proposed.

In our case, there is a (digital) unary-encoded up/down-counter at the heart of the controller, allowing to measure the phase difference between
  both clocks.
Note that since our goal is not to lock to a highly stable oscillator, our design is much simpler: our circuit only determines
  whether the actual phase offset is larger or smaller than the desired phase offset.
It is also worth noting that, while our oscillators are analog components, our circuit relies on the ability to
  switch between ``fast'' and ``slow'' only.
This binary decision may become un- or metastable frequently.
In stark contrast to a classical digital PLL with a binary counter, this does not pose a problem for our design.
We ensure that the potentially metastable output signal of our controller is only used to control the oscillator.
The oscillator frequency is required to remain in the range spanned by the frequencies possible under stable operation
  (slow and fast mode) in presence of a metastable, or in general, unstable signal.
This is the case for starved inverter ring oscillators.

The use of local clocks in our design has a further advantage over locking to a centralized clock that is assumed to provide a highly
stable frequency reference.
In our system the sender and receiver are not impaired by the failure of the respective other's clock.
While correct communication between the two nodes inherently requires both oscillators to work
correctly, our design guarantees that if one of the oscillators fails, the respective other keeps
running within the same frequency bounds.
Potential top-level error-detection based on the (non-)communicated data then provides
adequate application-specific reaction to such scenarios.

%###
\subsubsection*{Organization of the paper}
%###
We start with presenting the problem of communication in a system of two nodes with controllable oscillators
in Section~\ref{sec:systemmodel}.
We then break the system down into modules, formally specifying their requirements.
Section~\ref{sec:contTH} discusses gate-level implementations of the modules, together with proofs that the implementations
  satisfy the formal requirements.
In Section~\ref{sec:sims}, we present simulations of our implementation at gate-level (VHDL) and transistor level (Spice).
The simulation results are consistent with our formally proven results, and allow to obtain detailed
  performance metrics.
We conclude in Section~\ref{sec:conclusion}.
% END \input{intro.tex}
% BEGIN \input{model.tex}
\section{System Specification and Model}\label{sec:systemmodel}
We specify the system requirements and functionality next.
The link (see  \figref{fig:link}) has three parts: (i)~tunable oscillators $\Oscsnd$ and $\Oscrcv$,
(ii)~a (ring) buffer \buff, and
(iii)~a buffer controller \ctrl.
The link enables communication between two parties, a sender \snd\ and a receiver \rcv, that
   interact with the link via prescribed interfaces, discussed later on.

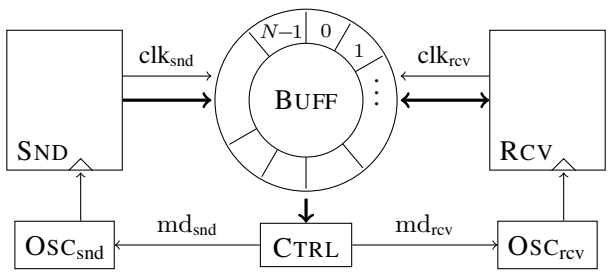
\begin{figure}\center
  \begin{tikzpicture}[scale=1.07,every node/.style={transform shape}]

  % sender and receiver
  \draw (0,0) node[rectangle,text width=1.2cm, text height=1.5cm,draw=black] (sender) {\snd};
  \draw (6,0) node[rectangle,text width=1.2cm, text height=1.5cm,draw=black] (receiver) {\rcv};

  % ring buffer
  \draw (3,0) node[circle,draw=black,inner sep=0.8cm] {};
  \node at (3,0) {\buff};
  \draw[-,black] (3,0) -- ++(90:1.1cm);
  \draw[-,black] (3,0) -- ++(-90:1.1cm);
  \draw (3,0) ++(74:0.9cm) node {\scriptsize $0$};

  \draw[-,black] (3,0) -- ++(130:1.1cm);
  \draw[-,black] (3,0) -- ++(-50:1.1cm);
  \draw (3,0) ++(110:0.9cm) node {\scriptsize $N\!\!-\!\!1$};

  \draw[-,black] (3,0) -- ++(60:1.1cm);
  \draw[-,black] (3,0) -- ++(-120:1.1cm);
  \draw (3,0) ++(43:0.9cm) node {\scriptsize $1$};
  \draw (3,0) ++(13:0.9cm) node {$\vdots$};

  \draw[-,black] (3,0) -- ++(30:1.1cm);
  \draw[-,black] (3,0) -- ++(-150:1.1cm);

  \draw (3,0) node[circle,draw=black,inner sep=0.5cm,fill=white] {};
  \node at (3,0) {\buff};

  % interconnect
  \draw[->,very thick] (sender.east) -- ++(1.1,0);
  \draw[<->,very thick] (receiver.west) -- ++(-1.1,0);

  \draw[->] (sender.east)++(0,0.3) -- ++(1.1,0) node[midway,above] {\small $\clksnd$};
  \draw[->] (receiver.west)++(0,0.3) -- ++(-1.1,0) node[midway,above] {\small $\clkrcv$};

  % osc
  \draw (0,-1.8) node[rectangle,text width=1cm, text height=0.3cm,draw=black] (oscS) {$\Oscsnd$\!};
  \draw (6,-1.8) node[rectangle,text width=1cm, text height=0.3cm,draw=black] (oscR) {$\Oscrcv$};

  % controller
  \draw (3,-1.8) node[rectangle,text width=0.9cm, text height=0.3cm,draw=black,align=center] (ctrl) {\ctrl};

  % interconnect
  \draw[->] (ctrl.west) -- (oscS.east) node[midway,above] {\small $\mdsnd$};
  \draw[->] (ctrl.east) -- (oscR.west) node[midway,above] {\small $\mdrcv$};

  \draw[->] (oscS.north)++(0.2,0) -- ++(0,0.6);
  \draw[->] (oscR.north)++(0.2,0) -- ++(0,0.6);

  \draw[-] (sender.south)++(0.05,0) -- ++(0.15,0.15) -- ++(0.15,-0.15);
  \draw[-] (receiver.south)++(0.05,0) -- ++(0.15,0.15) -- ++(0.15,-0.15);

  \draw[<-,very thick] (ctrl.north) -- ++(0,0.3);

\end{tikzpicture}
\caption{Link with Digital Controller}\label{fig:link}
\vspace{-0.3cm}
\end{figure}

The sender writes data to a ring buffer of even size $N > 0$, which is read by the receiver.
Cells are numbered from $0$ to $N-1$.
Read and write access is clocked: following transitions of its clock $\clksnd$,
the sender writes to the ring buffer. The register address is specified by the
current value of its address pointer, which it subsequently increments (modulo
$N$); likewise, following transitions of its clock $\clkrcv$, the receiver reads
from its current address and subsequently increments its pointer.

We remark that our design can easily be altered for bidirectional communication.
Each party needs to perform a read/write sequence instead of just a read (\rcv) (respectively
write (\snd)) operation when it is accessing a buffer cell;
the only effect is that the respective higher access time needs to be respected
in the timing constraints on the system. For ease of presentation, we stick to
the asymmetric setting in the following.

\subsection{Local Clocks}\label{sec:clockspec}%Clock Tuning

Sender and receiver clocks $\clksnd$ and $\clkrcv$ are derived from clock sources $\Oscsnd$ and $\Oscrcv$, respectively.
We require that these clock sources (or oscillators) are \emph{tunable in frequency} by the \emph{mode signals} $\mdsnd$ and $\mdrcv$.

Denote by $\discclk(t)\in \IZ$ a discrete clock value at wall-clock time $t\in \IR^+_0$.
This discrete clock is derived from a continuous clock $\realclk(t)\in \IR$ as $\discclk(t)\triangleq\lfloor \realclk(t) \rfloor$,
  with current frequency $\freqclk(t)$.
Let $\discclk_s(t),\discclk_r(t)$ be the discrete clock values of sender and receiver
  at wall-clock time $t$, and $\realclk_s(t), \realclk_r(t)$ their continuous clocks.
For properly chosen $\Tosc \geq 0$ and $\delta\leq 1$, we require:
\begin{enumerate}[(\bgroup\bfseries C1\egroup)]
    \item We assume that the clocks are started roughly at the same time:\footnote{For $\delta=1$, this is a fairly weak constraint. If sender and receiver each access one element of the ring buffer per clock cycle, it means that both oscillators are started within one clock cycle of each other. However, smaller values of $\delta$ may reduce the minimum feasible ring size by $2$ in some cases.} $\realclk_s(0),\realclk_r(0)\in (-\delta,0]$.
    \item If $\mdsnd$ ($\mdrcv$) is constantly $0$ during $[t-\Tosc,t]$, the sender (receiver) is in \emph{slow mode} at time~$t$ and $\freqclk_s(t)\in [s^-,s^+]$ ($\freqclk_r(t)\in [s^-,s^+]$).
    \item If $\mdsnd$ ($\mdrcv$) is constantly $1$ during $[t-\Tosc,t]$, the sender (receiver) is in \emph{fast mode} at time~$t$ and $\freqclk_s(t)\in [f^-,f^+]$ ($\freqclk_r(t)\in[f^-,f^+]$).
    \item If $\mdsnd$ ($\mdrcv$) is neither constantly $0$ nor constantly $1$ during $[t-\Tosc,t]$, the respective clock is \emph{unlocked} and $\freqclk_s(t)\in [s^-,f^+]$ ($\freqclk_r(t)\in [s^-,f^+]$).
    \item Clocks in slow mode are never faster than clocks in fast mode: $s^+\leq f^-$.
\end{enumerate}
Here, $\Tosc$ is the response time of the tunable oscillator. Note that our
requirements on the oscillator are fairly weak, making it easy to implement
(cf.~Section~\ref{sec:sims}): Only if the stable control signal is stable for
$\Tosc$ time, the oscillator needs to guarantee the respective rate. At any
other time, it is not locked to a fixed frequency mode and may run at any rate
between the slowest and fastest possible. This unlocked mode may be entered when
the control signal is ambiguous or transitioned recently, i.e., when both
parties are almost perfectly synchronized. The last condition is a minimal
requirement ensuring that the phase offset between the two clocks cannot
increase further when a clock in fast mode is chasing a clock in slow mode.

\subsection{Buffer Access Specification}\label{sec:bufferspec}%Buffer Access

Next, we specify buffer access in an abstract model with few parameters.
We assume that access to a buffer cell starts when the respective clock modulo $N$ (possibly with a fixed offset) equals the buffer index.
Note that this is a normalization of the time axis so that one computational cycle takes $1$ unit of ``local'' time as measured by the sender or receiver oscillator, respectively.\footnote{Note that this will typically not be $1$ unit of ``absolute'' time, as oscillator speeds may vary.} A computational cycle is defined by the local time between accessing consecutive buffers.

Intuitively, a buffer cell is \emph{valid} (i.e., ready to be read) if it contains stable, logical data and is currently not written.
A buffer cell is \emph{invalid} (i.e., ready to be written) if it is not valid and currently not read.
Formally:
\begin{enumerate}[(\bgroup\bfseries B1\egroup)]
    \item We define the receiver's (discrete) address pointer as $\rapointer(t)\triangleq \lfloor \rap(t) \rfloor \bmod N = \discclk_r(t)\bmod N$, where the receiver's (continuous) address pointer is $\rap(t) \triangleq \realclk_r(t)$.
    That is, the receiver starts to access cell $\ell$ at each time $t$ when $\rapointer(t)=\rap(t)\bmod N=\ell$.
    \item We define the sender's address pointer to be $\sapointer(t) \triangleq \lfloor \sap(t) \rfloor \bmod N$, where $\sap(t)\triangleq \realclk_s(t)+N/2$.
    That is, the sender pointer has a (nominal) offset of half the ring size relative to the receiver pointer.
    In the following, we will simply drop the ``starts to'' and say that the receiver (sender) \emph{accesses cell $\ell$ at time $t$} if $\rap(t)\bmod N = \ell$ ($\sap(t)\bmod N = \ell$).
    \item Read and write operations take non-zero time.
        We account for setup/hold times and latency by parameters $\tauwrt$ and $\taurd$, which denote the maximum ``durations'' of write and read operations.
        Concretely, if the sender accesses a cell at time $t$, the receiver must not do so during $[t,t+\tauwrt)$, and if the receiver accesses a cell at time $t$, the sender must not do so during $[t,t+\taurd)$.
    \item On initialization, cells $0\leq \ell < N/2$ are valid, while cells $N/2\leq \ell < N$ are invalid.
    If the sender accesses an invalid cell at time $t$, the cell \emph{becomes valid} at time $t+\tauwrt$.
    If the reader accesses a valid cell at time $t$, it \emph{becomes invalid} at time $t+\taurd$.
    This inductively defines for each cell and each time $t\geq 0$ whether it is valid or invalid.
\end{enumerate}
Note that these definitions are crafted in such a way that if the sender accesses only invalid cells and the reader accesses only valid cells, we have mutual exclusion of read and write operations and for each individual cell, reads and writes alternate.
This is the intended mode of operation, which we will formalize in Section~\ref{sec:correctdef}.

\begin{figure}[bt]
\centering
\begin{tikzpicture}[scale=1]

  % ring
  \draw ([shift=(65:2)]0,0) arc (65:-270:2);
  \draw[dotted] ([shift=(-270:2)]0,0) arc (-270:-280:2);

  \draw (0:1.9) -- ++(0:0.2) node[right] {\small 0}
    node[pos=1,xshift=20pt,yshift=-15pt] {\footnotesize $F_0(t) = 0$};

  \draw (-60:1.9) -- ++(-60:0.2) node[below,xshift=-3pt] {\small 1}
    node[pos=1,yshift=-5pt,xshift=25pt] {\footnotesize $F_1(t) = 0$};

  \draw (-120:1.9) -- ++(-120:0.2) node[below] {\small 2}
    node[pos=0,yshift=-25pt,xshift=-5pt] {\footnotesize $F_2(t) = 1$};

  \draw (-180:1.9) -- ++(-180:0.2) node[left] {\small 3}
    node[pos=1,yshift=-15pt,xshift=-20pt] {\footnotesize $F_3(t) = 1$};

  \draw (-240:1.9) -- ++(-240:0.2) node[above,xshift=3pt] {\small 4}
    node[pos=0,yshift=10pt,xshift=-30pt] {\footnotesize $F_4(t) = {\color{red}{ \metas}}$};

  \draw (-300:1.9) -- ++(-300:0.2) node[above,xshift=8pt] {\small $N-1$}
    node[pos=1,yshift=-3pt,xshift=35pt] {\footnotesize $F_{N-1}(t) = 0$};

  % p_r
  \draw[thick,->] (-65:1) -- ++(-65:0.5) node[pos=0,above,xshift=15pt] {\small $p_r(t-\tau_r)$};
  \draw[thick,->] (-105:1) -- ++(-105:0.5) node[pos=0,left,yshift=8pt,xshift=5pt] {\small $p_r(t)$};
  \drawsector[fill=red!30]{1.85}{0.3}{-105}{-40}{}
  \draw[->] ([shift=(-105:1.85)]0,0) arc (-105:-115:1.85);

  % p_s
  \draw[thick,->] (140:1) -- ++(140:0.5) node[pos=0,below] {\small $p_s(t-\tau_s)$};
  \draw[thick,->] (100:1) -- ++(100:0.5) node[pos=0,right] {\small $p_s(t)$};
  \drawsector[fill=red!30]{1.85}{0.3}{100}{-40}{}
  \draw[->] ([shift=(100:1.85)]0,0) arc (100:90:1.85);

  % P_r
  \draw[thick,->] (-60:2.7) -- ++(-60:-0.5) node[pos=0,below] {\small $P_r(t)$};

  % P_2
  \draw[thick,->] (120:2.7) -- ++(120:-0.5) node[pos=0,above] {\small $P_s(t)$};

\end{tikzpicture}
  \caption{Ring buffer access at time $t$. The sender currently accesses cell $4$.
    Hence, its full/empty flag is $\metas$.
    The receiver has just finished accessing cell $1$.
    Thus, its full/empty flag is $0$.
    In executions, we mark (potential) in-/metastability in red, cf.~\figref{fig:vhdl}.}\label{fig:buffer_access}
\end{figure}

\subsection{Metastability}

To minimize dead time of the control loop regulating the clock speeds, we do not make use of synchronizers.
Forgoing their use can result in meta-/unstable signals.
At any point in time, a signal has a value in $\{0,\metas,1\}$, where $\metas$ means that a signal is potentially metastable or
  in transition.
We employ a worst-case analysis, which assumes that $\metas$ propagates whenever possible;
  only explicit logical masking may protect from metastability, no probabilistic statements are used.

In particular, a flip-flop latching when its input is $\metas$ will ``store'' an $\metas$ until is latched again with a stable input.
Note that an output signal may also be unstable due to a transitioning signal, e.g.\ after latching a new value different from the previously stored one.

\subsection{Link Controller Interface Specification}\label{sec:contspec}%Controller

The mode signals themselves are generated by the controller \ctrl.
Controller decisions are based on full/empty flags of the ring-buffer cells, which we will describe shortly.
We stress that, inherently, the controller acts at the border of two clock domains.
Any digital implementation (including ours) is thus susceptible to metastable upsets.
Accordingly, the voltage levels of $\mdsnd$ and $\mdrcv$ may become meta-/unstable (between logical $0$ and $1$, denoted by $\metas$), as, in order to minimize delay, we do \emph{not} pipe them through a synchronizer
  chain before making use of them.

Let $\Tctr$ denote the maximum end-to-end delay of the controller circuit, i.e., between its input (the full/empty flags) and its output ($\mdsnd$ and $\mdrcv$).
The specification of the link controller's interface is as follows:
\begin{enumerate}[(\bgroup\bfseries L1\egroup)]
    \item If for $t\geq \Tctr$ the controller circuit specification maps the inputs during $[t-\Tctr,t]$ continuously to $1$ for signal $\mdsnd$, then $\mdsnd(t)=1$;
    analogous statements hold for output $0$ as well as signal $\mdrcv$ and outputs $0$ and $1$, respectively.
    \item In all other cases, the output at time $t$ is arbitrary, i.e., any value from $\{0,\metas,1\}$.
\end{enumerate}

\subsection{Full/Empty Flags}

With each buffer cell $\ell$, we associate a full/empty flag $F_{\ell}$.
It is specified as
\begin{enumerate}[(\bgroup\bfseries F1\egroup)]
\item $F_{\ell}(t)=1$ if the cell is valid at time $t$ and it either has not been accessed yet or the most recent access to it was by the sender;
\item $F_{\ell}(t)=0$ if the cell is invalid at time $t$ and it either has not been accessed yet or the most recent access to it was by the receiver;
\item if neither applies at time $t$, then $F_{\ell}(t)\in \{0,\metas,1\}$.
\end{enumerate}
In other words, we allow for the possibility that $F_{\ell}(t)=\metas$ at any point in time during read and write operations.

\figref{fig:buffer_access} depicts the state of the above described cell pointers
  at time $t$.
Observe that all cells between the sender and the receiver are full and thus their full/empty flags
  equal to to $1$, those between the receiver and the sender are empty with full/empty flags
  equal to $0$, and the flags of those currently accessed are $\metas$.
\subsection{System Correctness}\label{sec:correctdef}

Expressing the correct order of and separation in time between cell accesses,
we can now succinctly state what correct operation of the link architecture means.
\begin{definition}\label{def:correctsys}
A link is \emph{correct} if the following holds in any execution adhering to our model.
\begin{enumerate}[(\bgroup\bfseries P1\egroup)]
\item\label{P1} No underrun: the receiver accesses only valid cells.
\item\label{P2} No overflow: the sender accesses only invalid cells.
\end{enumerate}
\end{definition}
\begin{definition}\label{def:correctcont}
Controller \ctrl\ is \emph{correct } if it  computes the signals $\mdsnd$ and $\mdrcv$ out of the inputs $F_{\ell}$ so that the link is correct.
\end{definition}
The goal is now to design a (simple) controller that is correct even if the ring size $N$ is small: this minimizes both the size of the buffer and its latency.
% END \input{model.tex}
% BEGIN \input{continuous.tex}
\section{Continuous Threshold Controller}\label{sec:contTH}

Our control algorithm $\cont(\mathcal{T})$ is specified in Alg.~\ref{alg:controller}.
It is parametrized by $\thresh\in \IR^+$.
In the remainder of this section, we explain the intuition behind the approach.

For the purpose of exposition, denote by $\flevel(t) \triangleq \sap(t)-\rap(t) = N/2 + \realclk_s(t)-\realclk_r(t)$ the fill level of the buffer.
Recall that one of our design goals is to have a simple digital controller.
The most straightforward choice for such a control algorithm is presumably the threshold controller:
If the fill level of the ring buffer is larger than $N/2$, the sender is forced to slow and the receiver to fast mode.
If the fill level is less than $N/2$, the sender and receiver are forced into fast and slow mode, respectively.

However, as the various involved circuit components incur non-zero delays, we cannot expect instantaneous (and thus also not exact) information on the fill-level.
Also, changing the oscillators' speeds takes non-zero time, so we cannot hope for an immediate response to a small/large fill-level.
Alg.~\ref{alg:controller} takes this into account by introducing \emph{two thresholds}.
\begin{algorithm}[t]
\small
\begin{algorithmic}[1]
\TIME{}
  \STATE $\mdrcv(t) \gets$ choose arbitrarily in $\{0,\metas,1\}$
  \STATE $\mdsnd(t) \gets$ choose arbitrarily in $\{0,\metas,1\}$
  \IF{$\realclk_s(t)-\realclk_r(t) \ge \thresh$}
    \STATE $\mdrcv(t) \gets 1$ \hfill $//$ recall $\flevel(t) = N/2+\realclk_s(t)-\realclk_r(t)$
    \STATE $\mdsnd(t) \gets 0$
  \ENDIF
  \IF{$\realclk_r(t)-\realclk_s(t) \ge \thresh$}
    \STATE $\mdrcv(t) \gets 0$
    \STATE $\mdsnd(t) \gets 1$
  \ENDIF
\end{algorithmic}
\caption{Controller $\cont(\thresh)$}
\label{alg:controller}
\end{algorithm}
\figref{fig:nice} shows an execution where the controller \ctrl\ runs the algorithm.

\begin{figure}
  \begin{tikzpicture}
    % axis
    \draw[->] (0,-0.6) -- ++(0,2.3) node[above,xshift=0pt] {\small $\flevel(t)$};
    \draw[->] (0,-0.6) -- ++(7,0) node[right] {\small $t$};

    % bounds
    \draw[-,dotted,thick] (0,1.2) -- ++(6,0) node[right] {\small $\frac{N}{2}+\thresh$};
    \draw[-,dotted,thick] (0,0.0) -- ++(6,0) node[right] {\small $\frac{N}{2}-\thresh$};

    \draw[black] plot[smooth] coordinates
         {(0.2,0.2)
          (0.5,0.4)
          (0.8,0.6)
          (0.95,0.8)
          (1.1,1.0)
          (1.2,1.2)
          (1.6,1.4)
          (2.4,1.6)
          (2.9,1.6)
          (3.3,1.4)
          (3.9,1.2)
          (4.6,1.1)
          (5.0,1.4)
          (5.3,1.2)
          (5.6,1.0)
          (5.8,0.8)
          (6.3,0.6)
          (6.7,0.8)
          (7.0,1.0)
         };
    % controller
    \draw[<->] (1.2,-1) -- ++(0.8,0) node[midway,fill=white,xshift=0pt] {\scriptsize $\Tctr$}; % ends at 2
    \draw[<->] (4.8,-1) -- ++(0.8,0) node[midway,fill=white,xshift=0pt] {\scriptsize $\Tctr$}; % ends at 5.6

    % mode signal
    \draw (2,-1.125) rectangle ++(1.9,0.25) node[midway] {\scriptsize $1$};

    % mode labeling
    \draw (-0.55,-1.1) node[right,yshift=3pt] {\small $\mdrcv(t)$};

    % matching lines
    \draw[dotted,-] (1.2,-1.1) -- ++(0,2.2);
    \draw[dotted,-] (1.2,-1.1)++(2.7,0) -- ++(0,2.2);
    \draw[dotted,-] (4.8,-1.1) -- ++(0,2.2);
    \draw[dotted,-] (4.8,-1.1)++(0.5,0) -- ++(0,2.2);

    % osc labeling
    \draw (-0.55,-1.5) node[right,yshift=3pt] {\footnotesize $\Oscrcv$-mode};

    % osc
    \draw (2.9,-1.525) rectangle ++(1,0.25) node[midway] {\scriptsize fast};

    \draw[<->] (2,-1.4) -- ++(0.9,0) node[midway,fill=white,xshift=0pt] {\scriptsize $\Tosc$}; % ends at 2.9
    \draw[<->] (5.6,-1.4) -- ++(0.9,0) node[midway,fill=white,xshift=0pt] {\scriptsize $\Tosc$};
  \end{tikzpicture}
  \caption{$\cont(\thresh)$'s signals of the receiver.
  The fill-level increases until it hits $\frac{N}{2}+\thresh$, which makes the $\mdrcv$ signal drive $1$
    after $\Tctr$ time.
  After another $\Tosc$ time, the receiver and sender clocks are required to run in fast and slow mode,
    respectively (cf.\ Section~\ref{sec:systemmodel}).
  Note that the second phase during which the threshold $\frac{N}{2}+\thresh$ is crossed is too short for
    \ctrl\ and the oscillators to react with certainty.}\label{fig:nice}
  \vspace{-0.3cm}
\end{figure}

\subsection{Correctness of $\cont(\thresh)$}\label{sec:analysis}

Before we show that, for a $\thresh$ that is chosen sufficiently large, $\cont(\thresh)$ is \emph{implementable} by a digital circuit in Section~\ref{sec:clocked},
  we show that $\cont(\thresh)$ indeed is \emph{correct} (as per Definition \ref{def:correctcont}) if $\thresh$ is chosen small enough.
\begin{thm}\label{thm:correct}
$\cont(\thresh)$ is correct if
\begin{equation}\label{eq:T}
\begin{aligned}
\delta&\leq \thresh \\
&\leq N/2-(f^+-s^-)(\Tosc + \Tctr)-f^+ \max\{\tau_s,\tau_r\}.
\end{aligned}
\end{equation}
\end{thm}

 Recall that $\rap(t)=\realclk_r(t)$ and $\sap(t)=\realclk_s(t)+N/2$.
 Thus, when perfectly synchronized, the sender and receiver concurrently access opposite cells of the buffer.
 The first subtrahend accounts for the fact that the clocks remain unconstrained for $\Tosc+\Tctr$ time even after a threshold is reached:
 the controller guarantees corresponding output only after $\Tctr$ time, which is bound to affect clock speeds at most another $\Tosc$ time later;
 during this time period, one clock may ``catch up'' to the other at rate $f^+-s^-$.
 The second subtrahend accounts for the fact that the sender must always access a cell at least $\tau_r$ time before the receiver, while the receiver must do so $\tau_s$ time before the sender (B3).

 Note that these two conditions become fully symmetric when using $\max\{\tau_s,\tau_r\}$ as the minimum required separation between accesses.
 Translating this wall-clock time difference to the address pointers using the upper bound of $f^+$ on clock frequencies, we see that the following lemma is the key to showing Theorem~\ref{thm:correct}.
\begin{lem}\label{lem:bound}
If Eq.~\eqref{eq:T} holds, then
$$\forall t\in \IR^+_0\colon |\realclk_s(t) - \realclk_r(t)| \le N/2 - f^+ \max\{\tau_s,\tau_r\}\:.$$
\end{lem}
\begin{proof}
Assume for contradiction that $\realclk_s(t)-\realclk_r(t)>N/2 - f^+ \max\{\tau_s,\tau_r\}> \thresh$ for some time $t$.
Let $t_0 \in \IR^+_0$ be the minimal time such that $\realclk_s(\tau)-\realclk_r(\tau)\geq \thresh$ for all $\tau \in [t_0,t]$;
as $|\realclk_s(0)-\realclk_r(0)|<\delta\leq \thresh$ by (C1) and \eqref{eq:T} and both $\realclk_s$ and $\realclk_r$ are continuous, such a time $t_0$ must exist.
Observe that $\realclk_s(t_0)-\realclk_r(t_0)=\thresh$.

By the specification of the controller (L1), we have that $\mdsnd(\tau)=0$ and $\mdrcv(\tau)=1$ for all $\tau \in [t_0+\Tctr,t]$.
Thus, we have that $\freqclk_r(\tau)\geq f^-\geq s^+\geq \freqclk_s(\tau)$ for all $\tau\in [t_0+\Tctr+\Tosc,t]$ by the specification of the clocks ((C2), (C3), (C5), and (C6)).
Recall that also $\freqclk_r(\tau)\geq s^-$ and $\freqclk_s(\tau)\leq f^+$ at all times $\tau$ by (C2) to (C5).
If $t-t_0\geq \Tctr+\Tosc$, we can thus bound
\begin{align*}
&\realclk_s(t)-\realclk_r(t)
=\realclk_s(t_0)-\realclk_r(t_0)+\int_{t_0}^t \freqclk_s(\tau)-\freqclk_r(\tau)~d\tau\\
&\!\!\leq \thresh + \int_{t_0}^{t_0+\Tctr+\Tosc}f^+-s^-~d\tau
+ \int_{t_0+\Tctr+\Tosc}^{t}0~d\tau\\
&\!\!\leq \thresh + (f^+-s^-)(\Tctr+\Tosc)\stackrel{\eqref{eq:T}}{\leq}\frac{N}{2} - f^+ \max\{\tau_s,\tau_r\}.
\end{align*}
If $t-t_0< \Tctr+\Tosc$, the second part of the integral vanishes and the first part becomes smaller, showing that the same bound holds.
Either way, this contradicts our assumption that $\realclk_s(t)-\realclk_r(t)$ exceeds this bound.

Finally, we argue analogously for the case that $\realclk_r(t)-\realclk_s(t)>N/2 - f^+ \max\{\tau_s,\tau_r\}$, where the roles of sender and receiver are exchanged.
\end{proof}

 \begin{proof}[Proof of Theorem~\ref{thm:correct}]
 By Lemma~\ref{lem:bound},
 \begin{equation}\label{eq:no_overtake}
 \begin{aligned}
 &\quad\,|\sap(t)-\rap(t)|=|\realclk_s(t)+N/2-\realclk_r(t)|\\
 &\in[f^+ \max\{\tau_s,\tau_r\},N-f^+ \max\{\tau_s,\tau_r\}].
 \end{aligned}
 \end{equation}
 In particular, the (continuous) sender and receiver address pointers never have the same value modulo $N$ and thus cannot ``pass'' each other.
 Moreover, by our assumptions on the initial clock values~(C1), and since $\delta\leq 1$, we have that $\realclk_s(0),\realclk_r(0)\in (-1,0]$, i.e., $\rap(0)\in (-1,0]$ and $\sap(0)\in (N/2-1,N/2]$ by (B1) and (B2), respectively.
 Together with (B4), this implies that (i) the first access to each cell that is invalid at time $0$ is by the sender, (ii) the first access to each cell that is valid at time $0$ is by the receiver, and (iii) each cell is accessed alternatingly by sender and receiver.

 It remains to show that the receiver does not access a cell less than $\tau_s$ time after a sender access to the same cell.
 Similarly, we need to show that the sender does not access a cell less than $\tau_r$ time after a receiver access.
 To this end, suppose cell $\ell$ is accessed by the sender and receiver at times $t_s$ and $t_r$, respectively.
 Thus, $\ell = \rap(t_r)+aN=\sap(t_s)+bN$ for some $a,b\in \IZ$, i.e.,
 \begin{align*}
 |\rap(t_r)-\rap(t_s)|&=|\sap(t_s)-\rap(t_s)+(b-a)N|\\
 &\stackrel{\eqref{eq:no_overtake}}{\geq} f^+ \max\{\tau_s,\tau_r\}.
 \end{align*}
 As $\drap(t)=\dot{\realclk}_r(t)\leq f^+$ at all times $t$ by (C2) to (C5), we also have $|\rap(t_r)-\rap(t_s)|\leq f^+|t_r-t_s|$ and therefore
 $|t_r-t_s|\geq \max\{\tau_s,\tau_r\}$.
 Thus, (P1) and (P2) are satisfied for any access to cell $\ell$;
 since $\ell$ was arbitrary, this completes the proof.
 \end{proof}

\subsection{Clocked Implementation \clocked}\label{sec:clocked}

Next, we provide a simple and efficient controller implementation that works if $\thresh$ is sufficiently large.
Recall that our goal is to detect when $c_s(t)-c_r(t)\geq \thresh$ or $c_r(t)-c_s(t)\geq \thresh$.
By Lemma~\ref{lem:bound}, assuming a correct implementation satisfying \eqref{eq:T}, it holds that the address pointers never reach each other.
Together with the equality $c_r(t)-c_s(t) = p_r(t) + N/2 - p_s(t)$, it follows that all we need to check is whether one pointer is more or less than $N/2$ cells ``ahead'' of the other or not.
This gives us an indication of whether the buffer is more or less than half full, and the more accurately we can decide, the smaller $\thresh$ can be for the implementation to be correct.

We use the receiver's clock to sample whether the sender's address pointer is
currently by more or less than $N/2$ cells ahead of the receiver's address pointer.\footnote{It is worth noting that one could use a purely combinational controller to achieve the same result, i.e., there is no need to rely on clocking. Making use of the clock does also not guarantee that stable values are sampled. However, making use of the clock results in a controller with smaller threshold value than a straightforward combinational implementation due to the known alignment of the sampling times with one of the clocks.}
This is where the full/empty flags come in handy.
Instead of having to communicate and sample $c_s(t)$, the receiver simply samples the flag of cell $\ell + N/2 \bmod N$ when accessing cell $\ell \in [N]$.
This occurs at each time $t$ when $\ell = p_r(t) \bmod N = c_r(t) \bmod N$, which means that if the buffer is exactly half full, we had that $p_s(t) \bmod N = \ell + N/2 \bmod N$, i.e., the sender accesses cell $\ell + N/2 \bmod N$ at precisely the same time.
This means that it starts setting the full/empty flag of the cell from $0$ to $1$ at time $t$, i.e., if the buffer is less than half full, the receiver will successfully sample a stable $0$ into flip-flop ffa,\footnote{For simplicity, we attribute any unstable reading to the transition of the memory flag of the cell via $\tau_s$. However, of course the parameters of the flip-flop we sample into, quality of the clock signal, and the delay from the flag's output to the flip-flop's input through the MUX all have an effect. Based on a timing analysis of the circuit and adding a suitable phase shift to the clock input of ffs by, e.g., using a buffer, the abstract behavior we assume can be realized. $\tau_s$ then simply describes the size of the time window during which ffs is vulnerable to metastability induced by a transition of the memory flag of cell~$\ell$.} see \figref{fig:control_clocked}.

In contrast, if the buffer is more than half full, it may be the case that the receiver ``reads'' an $\metas$, because the sender is still writing the full/empty flag.
Only if it accessed the cell at the latest at time $t-\tau_s$, we can be certain that the result of the read operation is a stable $1$.
To avoid this asymmetry, we sample cell $\ell$ at times $t$ when $c_r(t) \bmod N = \ell + f^+\tau_s/2$.
\begin{lem}\label{lem:clocked_ctrl}
Suppose time $t$ and $\ell \in [N]$ are such that $c_r(t) \bmod N = \ell + f^+\tau_s/2$ and Eq.~\eqref{eq:no_overtake} holds. Then
\begin{enumerate}[(i)]
  \item $c_s(t)-c_r(t)\geq f^+\tau_s/2\Rightarrow F_{\ell}(t)=1$, and
  \item $c_r(t) - c_s(t)\geq f^+\tau_s/2\Rightarrow F_{\ell}(t)=0$.
\end{enumerate}
\end{lem}
\begin{proof}
We show (i) first, i.e., assume that $c_s(t)\geq c_r(t) + f^+\tau_s/2$.
Then
\begin{equation*}
p_s(t-\tau_s)=c_s(t-\tau_s)\geq c_s(t)-f^+\tau_s \geq c_r(t) - f^+\tau_s/2\,.
\end{equation*}
Note that
\begin{equation*}
c_r(t) -f^+\tau_s/2 \bmod N = \ell\,,
\end{equation*}
i.e., the sender completed writing cell $\ell$ (for the most recent time) at time $t$;
here, \eqref{eq:no_overtake} shows that neither sender nor receiver cannot have accessed the cell again after the operation was complete.
In other words, $F_{\ell}(t)=1$, as claimed.

Now we show (ii).
Thus, we assume that $c_s(t)\leq c_r(t) - f^+\tau_s/2$, while also
\begin{equation*}
c_r(t) - f^+\tau_s/2 \bmod N = \ell\,.
\end{equation*}
Hence, the most recent access to cell $\ell$ was by the reader (again using also \eqref{eq:no_overtake}), which also completed its access (as $N\geq 2$ and we assume that operations are completed within a single clock cycle).
In other words, $F_{\ell}(t)=0$, as claimed.
\end{proof}

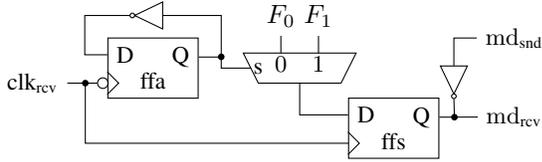
\begin{figure}[bt]
  \centering
  \begin{tikzpicture}[circuit logic US]

    % FFA
    \draw[draw] (0,0) node[draw=black,rectangle,minimum width=1.2cm,minimum height=.8cm] (ff1) {};
    \draw (ff1.west)++(-0.06,-.2) circle (2pt);
    \draw (ff1.north east) node[below left] (Q1) {\small Q};
    \draw (ff1.north west) node[below right] (D1) {\small D};
    \draw (ff1.south) node[above] (ffa-label) {\small ffa};
    \draw[] (ff1.west)++(0,-0.35) -- ++(0.15,0.15) -- ++(-0.15,0.15);

    % \node[inverter, fill=white, draw, rotate=-90, scale=0.3] (inv1) at (0,.7) {};
    \node[not gate, point left, scale=.5] (inv1) at (0,.7) {};
    \draw[] (Q1.east) -- ++(right: .3) |- (inv1.input);
    \draw[] (D1.west) -- ++(left: .3) |- (inv1.output);

    \draw[] (ff1.east)++(.3,.15) node[circle,fill=black,inner sep=0pt,minimum size=2.5pt] {} |- (1.5,0);

    % mux
    \draw[fill=white] (1.2,.2) -- ++(1.5,0) -- ++(-0.2,-0.4) -- ++(-1.1,0) -- cycle
      node[xshift=.5cm,yshift=-4pt] {\small $0$}
      node[xshift=1cm,yshift=-4pt] {\small $1$}
      node[xshift=.2cm,yshift=-.2cm] {\small s};

    \node (full0) at (1.7,.7) {$F_0$};
    \draw[] (full0) -- ++(down: .5);
    \node (full1) at (2.2,.7) {$F_1$};
    \draw[] (full1) -- ++(down: .5);

    % FFS
    \draw[draw] (3.2,-0.8) node[draw=black,rectangle,minimum width=1.2cm,minimum height=.8cm] (ff2) {};
    \draw (ff2.north east) node[below left] (Q2) {\small Q};
    \draw (ff2.north west) node[below right] (D2) {\small D};
    \draw (ff2.south) node[above] (ffa-label) {\small ffs};
    \draw[] (ff2.west)++(0,-0.35) -- ++(0.15,0.15) -- ++(-0.15,0.15);

    \draw[] (D2.west) -| (1.95,-.2);

    \node (clk) at (-1.6,-.2) {\small $\clkrcv$};
    \draw[] (clk) -- (-.74,-.2);
    \draw[] (clk) ++(right: .7) node[circle,fill=black,inner sep=0pt,minimum size=2.5pt] {} |- (2.6,-1);

    \draw[] (Q2.east) -- (4.8, -.65) node[fill=white] {\small $\mdrcv$};

    % \node[inverter, fill=white, draw, rotate=180, scale=0.3] (inv2) at (4,-.1) {};
    \node[not gate, point down, scale=.5] (inv2) at (4,-.1) {};
    \draw[] (Q2.east)++(right: .2) node[circle,fill=black,inner sep=0pt,minimum size=2.5pt] {} -- (inv2.output);
    \draw[] (inv2.input) |- (4.8, .4) node[fill=white] {\small $\mdsnd$};

  \end{tikzpicture}
  \caption{Controller \clocked\ for ring-size $N=2$. Flip-flop ffa stores the address (modulo $2$) that is sampled and
    ffs the sampled full/empty flag.}\label{fig:control_clocked}
\end{figure}

Based on this idea, we derive a straightforward implementation of the controller.
Put simply, the receiver samples the full/empty flag of the cell opposite to the one it currently reads in the ring.
More precisely, $\mdrcv$ is the output of a flip-flop (flip-flop `ffs' in  \figref{fig:control_clocked}), into which the receiver samples $F_{\ell}(t)$ at times $t$ such that $c_r(t) \bmod N = \ell + f^+\tau_s/2$.
Signal $\mdsnd$ is obtained by negating $\mdrcv$.
A circuit implementing this approach for ring size $N=2$ is shown in \figref{fig:control_clocked}.
Here, flip-flop ffa is a modulo $2$ counter used to track the address to the current cell to sample.
It is initialized to the opposite of the receiver address.
We need to ensure that the MUX switches to forwarding the respective full flag before
flip-flop ffs latches the output of the MUX. We do so by computing the select bit
on the negated clock signal. This shifts the computation of the select bit
by half a clock cycle and ensures correct timing. Note that here we might get
metastable mode signals due to switching full flags. %
% As by our assumption the
% voltage controlled oscillator operates in the defined frequency range the proof ensures correct behaviour of the link.
Naturally, it is necessary that the mode signal is computed within a single clock cycle;
given the simplicity of the circuit, this is easily achieved.

In the following, denote by $\tau_{\max}$ the maximum propagation time through the
circuit shown in \figref{fig:control_clocked} from the full/empty flags at the
top to $\mdsnd$ (without $\tau_s$, which is already taken into account by
Lemma~\ref{lem:clocked_ctrl}). Lemma~\ref{lem:clocked_ctrl} then characterizes the
proposed controller.

\begin{cor}\label{cor:clocked_ctrl}
Assume that the control circuit \clocked\ is used in accordance with Lemma~\ref{lem:clocked_ctrl} and that (P1) and (P2) hold until time
  $t > \Tctr = 1/s^-+\tau_{\max}$.
\begin{enumerate}[(i)]
  \item If for all $t'\in [t-\Tctr,t]$ we have that
  \begin{equation*}
  \realclk_s(t')-\realclk_r(t')\geq f^+\tau_s/2\,,
  \end{equation*}
  then $\mdrcv(t)=1$ and $\mdsnd(t)=0$.
  \item If for all $t'\in [t-\Tctr,t]$ we have that
  \begin{equation*}
  \realclk_r(t')-\realclk_s(t')\geq f^+\tau_s/2\,,
  \end{equation*}
  then $\mdsnd(t)=1$ and $\mdrcv(t)=0$.
\end{enumerate}
\end{cor}
\begin{proof}
The outputs $\mdrcv(t)$ and $\mdsnd(t)$ at time $t$ are derived from the output of ffs at time $t$ (or one inverter delay earlier).
As the receiver clock runs at least at speed $s^-$ (by (C2)--(C4)), flip-flop ffs is latched at least every $1/s^-$ time.
Hence, taking into account the propagation time through the MUX and the definition of $\tau_{\max}$, the outputs correspond to the output of one of the flags at some time $t'\in [t-\Tctr,t]$.
As the MUX selects the flag output it forwards according to Lemma~\ref{lem:clocked_ctrl}, we can apply the lemma to time $t'$, yielding in Case~(i) that a stable $1$ is latched and in Case~(ii) that a stable $0$ is latched.
This results in the desired corresponding circuit outputs $\mdrcv(t)=1$ and $\mdsnd(t)=0$ (Case~(i)) or $\mdsnd(t)=1$ and $\mdrcv(t)=0$ (Case (ii)), respectively.
\end{proof}

We now can derive the correctness of the controller, expressed in Theorem~\ref{thm:clocked_correct}, conditional on simple constraints on~$\thresh$.

\begin{thm}\label{thm:clocked_correct}
Assume that Eq.~\eqref{eq:T} holds, where $\Tctr=1/s^-+\tau_{\max}$, and $\thresh \geq f^+\tau_s/2$.
Then \clocked\ is an implementation of $\cont(\thresh)$.
\end{thm}
\begin{proof}
If there is some access to a valid cell by the sender or to an invalid cell by the reader, there must be a minimal such time (because the start of a cell access is a discrete event).
Denote by $\bar{t}$ the minimal such time if such an access occurs and set $\bar{t}$ to infinity otherwise.

We claim that the circuit implements $\cont(\thresh)$ at all times $0\leq t<\bar{t}$; from this we will infer the statement of the theorem.
Recall that by (L1) and (L2), the controller implementation needs to output a specific (and stable) signal only if the
condition in Line~3 or the one in Line~7 of Algorithm~\ref{alg:controller} continuously holds during the previous $\Tctr$ time.
According to Algorithm~\ref{alg:controller}, this is the case at time $t$ if and only
if $c_s(t')-c_r(t')\geq \thresh$ for all $t' \in [t-\Tctr,t]$ or $c_r(t')-c_s(t')\geq \thresh$ for all $t'\in [t-\Tctr,t]$.

Consider such a time $t$.
Note that $t>\Tctr$, as $|c_s(0)-c_r(0)|<\delta\leq \thresh$ by (C1) and Eq.~\eqref{eq:T}, i.e., neither condition is satisfied at time $0$.
We consider the two cases (i) $\realclk_s(t')-\realclk_r(t')\geq \thresh$ for all $t'\in [t-\Tctr,t]$ and (ii) $\realclk_r(t)-\realclk_s(t)\geq \thresh$ for all $t'\in [t-\Tctr,t]$.

 \medskip

 {\bf Case (i):} Since $\thresh \ge f^+\tau_s/2$, we may apply Case~(i) of Corollary~\ref{cor:clocked_ctrl}.
 We conclude that $\mdrcv(t)=1$ and $\mdsnd(t)=0$.

 \medskip

 {\bf Case (ii):} In this case we may apply Case~(ii) of Corollary~\ref{cor:clocked_ctrl}, from which we deduce that
 $\mdrcv(t)=0$ and $\mdsnd(t)=1$.

\medskip

We conclude that the circuit meets the specification at all times $t<\bar{t}$.
In particular, we can apply Lemma~\ref{lem:bound} at times $t<\bar{t}$,
  showing that $|c_s(t)-c_r(t)|\leq N/2-f^+\max\{\tau_s,\tau_r\}$.
If $\bar{t}\neq \infty$, continuity of $c_s$ and $c_r$ implies that also
  $|c_s(\bar{t}\,)-c_r(\bar{t}\,)|\leq N/2-f^+\max\{\tau_s,\tau_r\}$.
Reasoning analogously to the proof of Theorem~\ref{thm:correct}, it follows that (P1) and (P2) are not violated at
  times $t\leq \bar{t}$, contradicting the definition of $\bar{t}$.
We conclude that $\bar{t}=\infty$, implying that the circuit from \figref{fig:control_clocked}
  indeed implements $\cont(\thresh)$.
\end{proof}

Finally, we translate the theorem into a sufficient condition for correctness of the link implementation.
To state its performance, we define the \emph{latency} as the maximum time between consecutive accesses of the sender and receiver to the same cell, plus the setup/hold time at the receiver (as the data should be stable before it is used).
The \emph{throughput} is the guaranteed minimum rate of delivered packets; note that no packet drops or corruptions occur in our implementations.
\begin{cor}\label{cor:clocked_correct}
For $\Delta =\lceil(f^+-s^-)(\Tosc+1/s^-+\tau_{\max})+ f^+\max\{\tau_s,\tau_r\}+\max\{\delta,f^+\tau_s/2\}\rceil$
and $N\geq 2\Delta$, the given clocked link implementation is correct with latency $N/s^-$ and throughput $1/s^-$.
\end{cor}
\begin{proof}
Set $\Tctr =1/s^-+\tau_{\max}$.
We choose $\thresh$ such that \eqref{eq:T} and $\thresh\geq f^+\tau_s/2$ are both satisfied.
This is possible if and only if $N/2\geq \Delta$, which holds by the prerequisites of the corollary.
Then Theorem~\ref{thm:clocked_correct} yields that the circuit from \figref{fig:control_clocked} indeed implements $\cont(\thresh)$, and Theorem~\ref{thm:correct} shows that the implemented controller is correct.

The performance bounds follow immediately from correctness and the fact that the guaranteed minimum clock rate is~$s^-$.
\end{proof}
% END \input{continuous.tex}
% BEGIN \input{simulations.tex}
\section{Performance Evaluation}\label{sec:sims}

We discuss a UMC 65\,nm ASIC design operating at roughly $2$\,GHz for which we
carried out simulations. This demonstrates that the derived performance bounds
indeed lead to promising results.

In the section, we also demonstrate simulated executions
  that show the circuit behaving according to the specification,
  despite reoccurring metastability of its control signals;
  see \figref{fig:vhdl}.
In fact metastability of the control signals
  is likely to be observed in an implementation, since by its attempt
  to synchronize the two oscillators, the controller repeatedly drives
  the control signals into metastability; much like experimental
  setups to measure deep metastability of
  synchronizers \cite{zhou2008chip,polzer2013approach}.
We would like to point out that any such demonstration, however, does not
  replace the correctness proofs in Section~\ref{sec:contTH}.
Proving that metastability is not a problem would require to
  verify the absence of metastability (or resulting effects)
  in all circuit components, except for the places to which our proofs
  show metastability to be confined.

\subsection{ASIC Implementation}

\begin{figure*}
  \begin{center}
    \includegraphics[width=.7\textwidth]{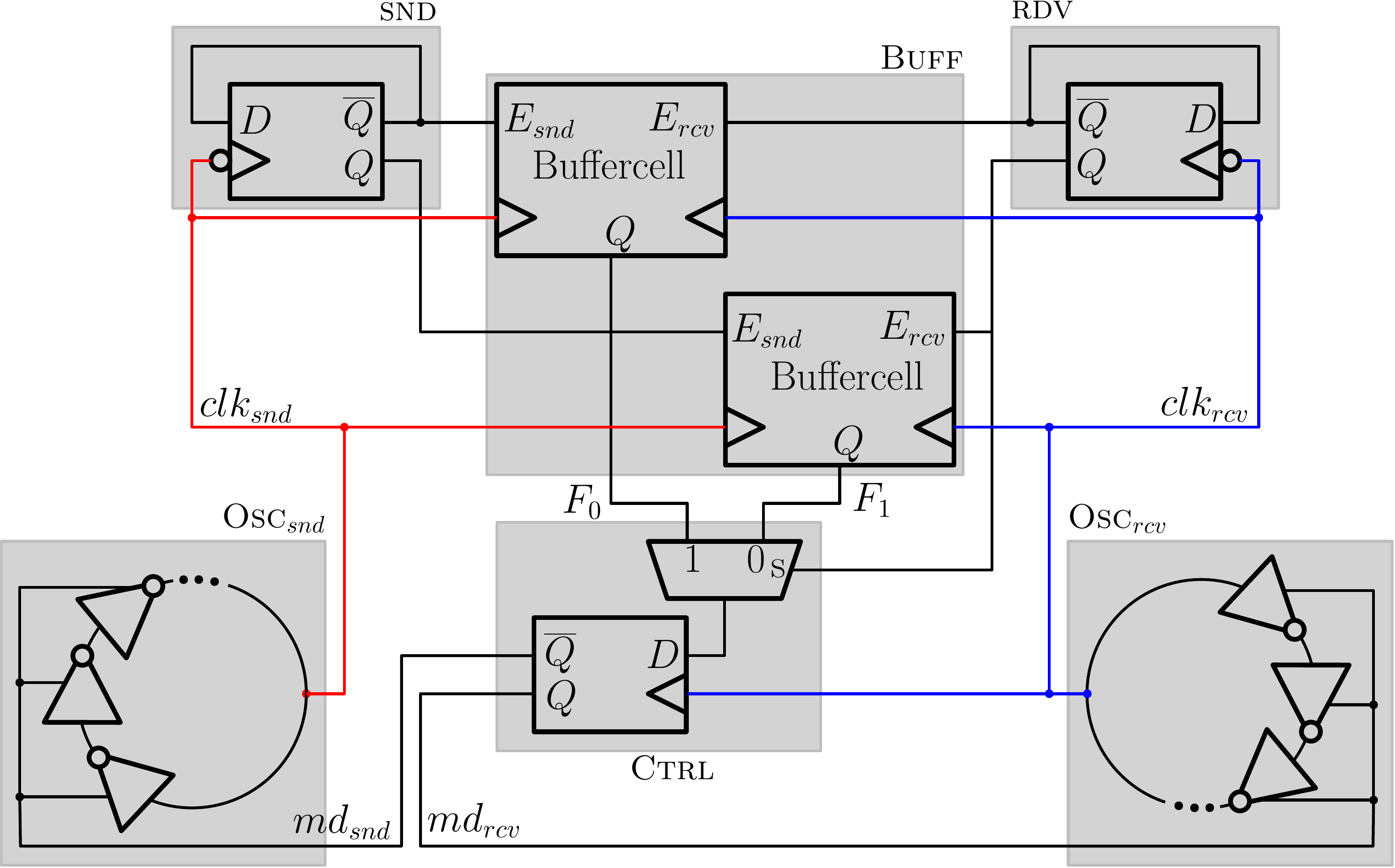}
  \caption{Implementation of the system with buffer size $N=2$. Clock regions are
  marked red (sender) and blue (receiver).}
  \label{fig:system}
  \end{center}
\end{figure*}

The complete design is shown in \figref{fig:system}. It comprises the digital controller ($\ctrl$),
tunable sender and receiver oscillator ($\Oscsnd$,$\Oscrcv$), and the ring buffer of size $N=2$.

\begin{figure}
  \begin{center}
  \includegraphics[height=4cm]{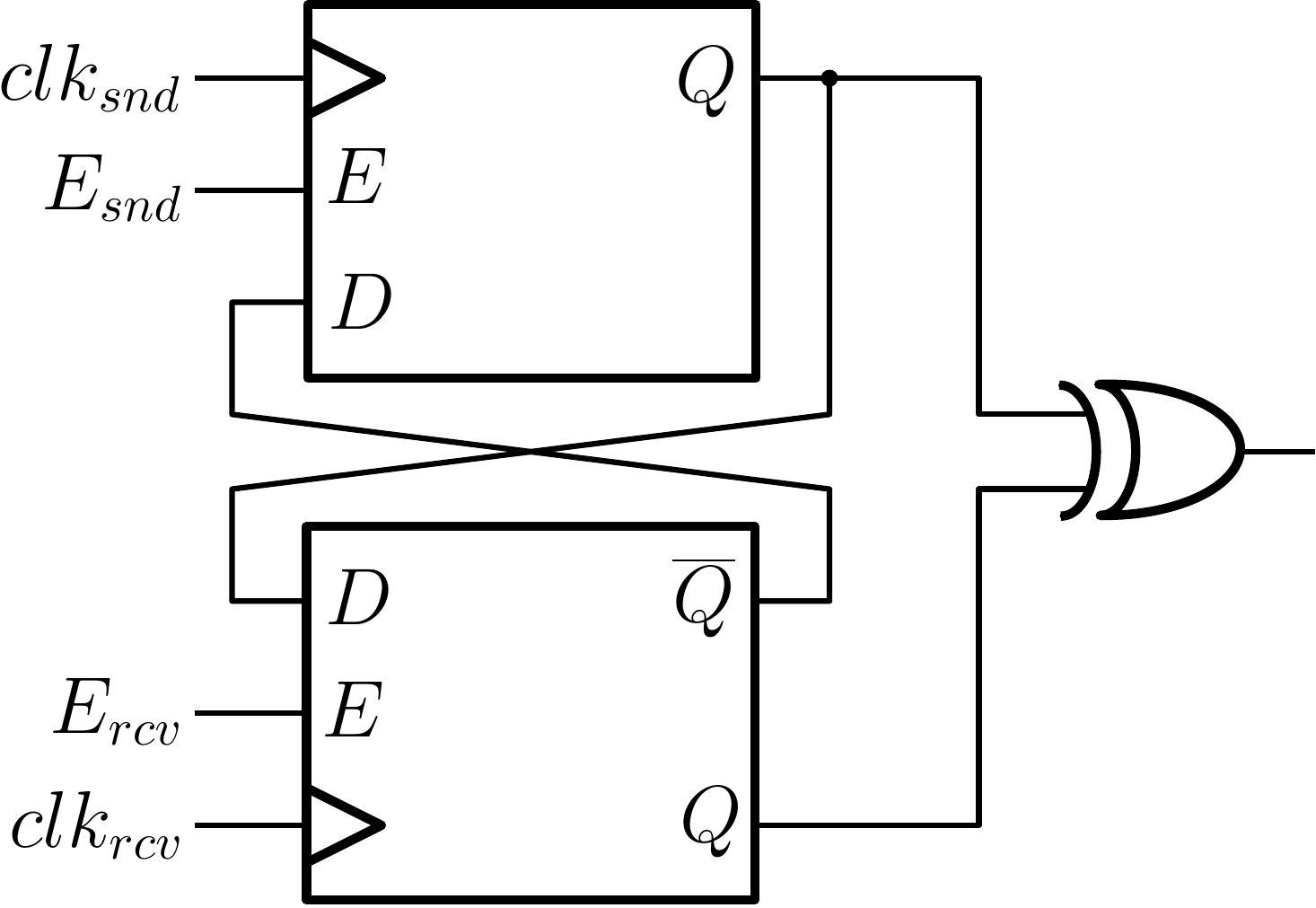}
  \caption{Implementation of a buffer cell that can only be set by the sender
  and only be reset by the receiver.}
  \label{fig:buffer cell}
  \end{center}
\end{figure}

At a buffer size of $2$ the address logic in $\snd$ and $\rcv$ reduces to a simple
modulo $2$ counter. Hence, we only have a single register for the sender and the
receiver side. The modulo counter operates on the negated clock to ensure
a stable output at the time a register in the buffer is accessed. The buffer consists of two buffer cells that store the
full/empty-flags. The design of a buffer cell that can be set to $1$ by one clock
domain and reset to $0$ by another clock domain is given in \figref{fig:buffer cell}.
The design uses a flip-flop for each clock domain that forwards its output to a
XOR which computes the output. If the sender flip-flop is enabled it copies the
negation of the receiver state. For differing states the XOR will output a $1$.
If the receiver enables its flip-flop the state of the sender is copied. Hence,
the output of the buffer cell is reset to $0$.

We can optimize the controller from \figref{fig:control_clocked}, as we already
compute the write address of the sender. We remove flip-flop ffa and
read the address from the $\snd$ address logic. The multiplexer in $\ctrl$
is connected such that we sample from the buffer cell that is currently not written by the sender.
The timing diagram in \figref{fig:timing_det} shows the behavior of $\ctrl$.

\begin{figure}
  \begin{center}
  \includegraphics[width=\columnwidth]{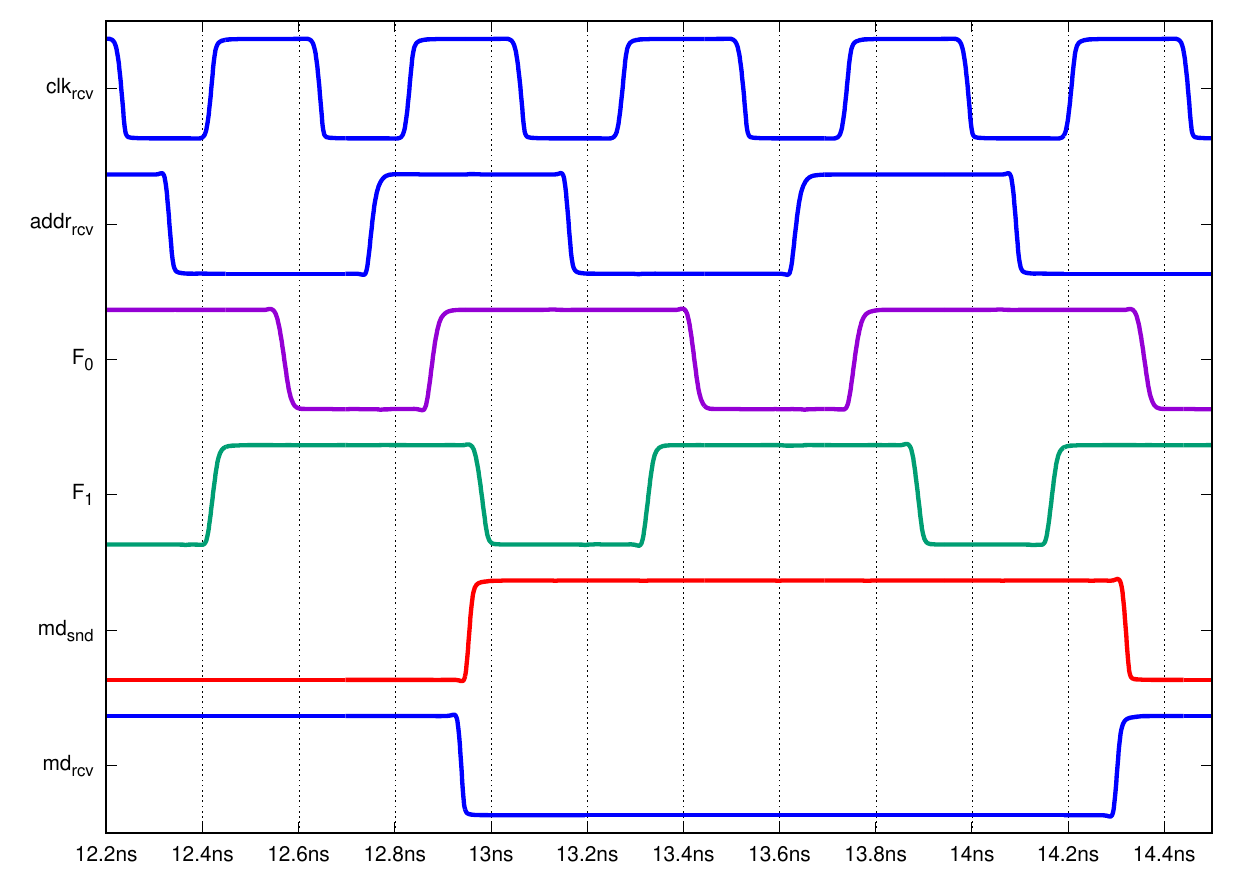}
  \caption{Timing diagram of the controller $\ctrl$. The address $addr_{rcv}$
  decides which full-flag is sampled into the register of the controller at a
  rising clock transition.}
  \label{fig:timing_det}
  \end{center}
\end{figure}

% Due to being purely digital, the {\em controller\/} and {\em ring buffer\/} VHDL design entries
%   and synthesis (Cadence Encounter RTL Compiler to UMC 65\,nm) followed a standard approach,
%   except that we manually checked that optimization preserves the circuit's containment properties. \joh{no more}

Recall that we require the sender and receiver {\em oscillators\/} to be well-behaved even when control bits are unstable.
Specifically, we require that (i) oscillator frequencies are always within
  $[s^-,f^+]$, and (ii) frequency mode changes occur within $\Tosc$ time ((C2) to (C5)).
This is why we resorted to starved-inverter ring oscillators that guarantee such
  behavior~\cite{hajimiri1999jitter}; we designed the sender and receiver starved inverter rings
  at transistor level following \cite{suman2016analysis}.
% In their design a delay cell is a starved inverter with 6 transistors.
Note that we do not need the full control logic overhead typically required to drive the starved inverter cells,
since we only need two speeds: slow and fast.
Hence, the control logic of the oscillator takes a single bit and adjusts the delay of the starved
inverters according to fast or slow mode.
As the receiver oscillator $\Oscrcv$ additionally drives the control logic $\ctrl$
its load is higher than the load driven by the sender oscillator. The
effect is that $\Oscrcv$ has slightly slower fast and slow modes. The difference
does not matter as long as the oscillator speeds lie within their theoretical
bounds. One can keep the imbalance very small by decoupling the oscillators from the load
with buffers.

Signals $\mdrcv$ and $\mdsnd$ are used as control signals of the rings, which run at roughly $2$\,GHz
  and $2.3$\,GHz for input $0$ and $1$, respectively.

Extracting delay and frequency parameters from the standard cell library we get
$\Delta=1$ in Corollary~\ref{cor:clocked_correct}, i.e., \clocked\ is provably correct for $N\geq 2$.
This fits to the bounds given in Table~\ref{table:circbounds}.

\subsection{Frequency Stability of Tunable Oscillators}

Typically, accuracy of oscillator frequencies is stated as a two-sided error,
i.e., if the nominal frequency of the oscillator is $f$ and it has a relative
frequency error of at most $r$, then at any time its momentary frequency is between $(1-r)f$ and $(1+r)f$.

Recall from (C5) that we require that the fast oscillator mode is always faster than the slow oscillator mode.
For a $2-2.3$\,GHz clock we must therefore tune the clock within an error $r$ that satisfies the condition $2\cdot(1+r)^2/(1-r)^2\leq 2.3$,
i.e., $r\leq 3.49$\% is a sufficient bound on the frequency error.
In case these error margins would be too restrictive, we could choose a clock with larger gap between fast and slow modes, e.g.,
$2-2.5$\,GHz.
Depending on the outcome of the timing analysis (see also Corollary~\ref{cor:clocked_correct}), this may require a larger buffer size $N$.

For comparison, the accuracy requirements for the oscillators used in \cite{dally2010even} are as follows.
If both the sender and receiver oscillator run at (roughly) the same nominal frequency, $\Delta p < g/S$ is proven to be sufficient for
correctness of the design, where $\Delta p$ is the relative phase change per clock cycle,
$S=4$ the number of synchronizer stages, and $g=0.1$ the guard band.
However, the proof assumes a perfectly stable receiver clock.
If receiver and sender oscillator may drift, the above inequality becomes $2 \Delta p (1+ \Delta p)<g/S$.
This is equivalent to a frequency error of less than $1.24$\%.

\subsection{Gate level and SPICE Simulations}

We first ran gate-level VHDL simulations of designs of our \clocked\ controller with delay and setup/hold parameters from the ASIC design.
The starved-inverter rings were simulated by forward Euler integration of a first order ODE model, where current clock rates
  are independently uniformly distributed in each integration step to account for drift.
The high respectively low frequency of the starved inverter rings where set to $2.3$\,GHz
respectively $2$\,GHz.
Potential in-/metastability of signals was simulated by \texttt{X} in a worst-case manner; this includes flip-flops
  with setup/hold violations, full/empty-flags, and oscillator mode signals.
Simulated traces were $5$\,ms ($10^7$ clock cycles) long and all in accordance with the proven correctness results.
We stress that signals $\mdrcv$ and $\mdsnd$ were unstable (\texttt{X}) almost all the time due to the conservative
  gate model assumptions, yet no buffer over-/underruns were encountered; cf.\ \figref{fig:vhdl}.

\begin{figure*}
  \includegraphics[width=1\textwidth]{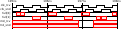}
  \vspace{-0.6cm}
  \caption{Gate-level simulations for link with \clocked.}
  \vspace{-0.3cm}
  \label{fig:vhdl}
\end{figure*}

We then ran Spice simulations for the \clocked\ design:
The design was implemented in Spice using standard cells and parameters of
the UMC 65\,nm library combined with an implementation of a tunable ring
oscillator. The oscillator runs at speed $2.09$\,GHz in slow mode and
$2.42$\,GHz in fast mode. Taking into account timing constraints and propagation
delays of the elements we can use a ring buffer of size two, according
to Corollary~\ref{cor:clocked_correct}.

When simulating the design for $500$\,ns (about $1100$ clock cycles) no faulty
behavior could be detected. However, the simulation confirms what we stressed
previously. In almost $50\%$ of the cases the setup time of $\operatorname{ffs}$ (see
\figref{fig:control_clocked}) is violated due to late transition of the
full flags. Still the controller behaves correct and the two oscillators run
synchronously.
% However, for the simulated time period this did not result in metastability.

% We initialize according to (B4) and set the pointers of sender and
% receiver side accordingly. The design stabilizes in roughly $10$ clock cycles to an
% equilibrium state where mode signals oscillate with a frequency of $0.57$\,GHz
% (see \figref{fig:md}).

% We show an excerpt of $4$\,ns from the simulation.
\figref{fig:full} shows
the full flags of the buffer. Sender and
receiver alternatingly access cell $0$ and cell $1$. \figref{fig:clk}
shows the clock signals produced by the sender and receiver oscillators. When
stabilized, the sender is ahead by slightly more than a clock cycle. Both run on average with a
frequency of roughly $2.28$\,GHz. \figref{fig:md} shows the mode signals of
the sender and receiver which are computed by \clocked.

\begin{figure*}
\centering
\begin{subfigure}[b]{\textwidth}
  \includegraphics[width=\textwidth,clip,trim=6 103 3 100]{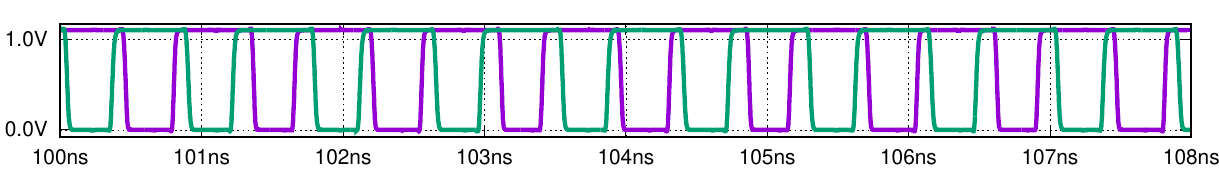}
  \caption{full flags}
  \label{fig:full}
\end{subfigure}
\begin{subfigure}[b]{\textwidth}
  \includegraphics[width=\textwidth,clip,trim=6 103 3 100]{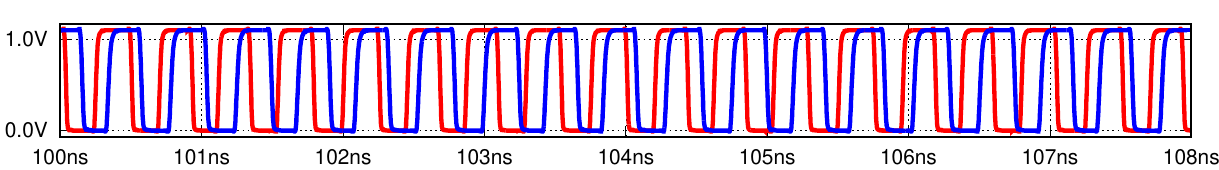}
  \caption{clock signals}
  \label{fig:clk}
\end{subfigure}
\begin{subfigure}[b]{\textwidth}
  \includegraphics[width=\textwidth,clip,trim=6 103 3 100]{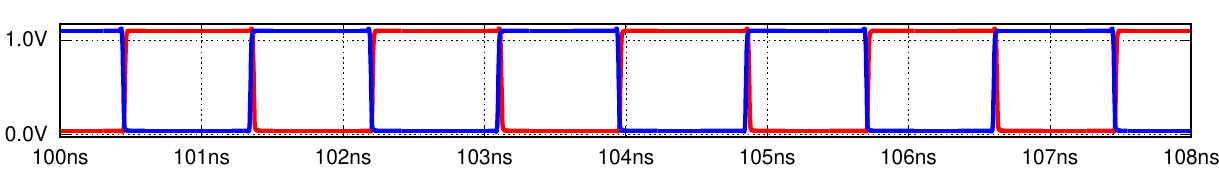}
  \caption{mode signals}
  \label{fig:md}
\end{subfigure}
\caption{Ring buffer with two cells. (a) Rising and falling full flags of cell $0$ (purple)
  or $1$ (green) show write and read access to the respective cell. (b) Clock signals of the sender (red) and receiver (blue) oscillator. When stabilized
  both run at $2.28$\,GHz on average. (c) The mode signals for sender (red) and receiver (blue) side alternate between
    fast ($2.42$\,GHz) and slow ($2.09$\,GHz) mode.}
\vspace{-0.5cm}
\end{figure*}

\subsection{Increasing Initialization Slack}\label{sec:slack}

If a sufficiently small $\delta$ (i.e., initial clock offset) cannot be guaranteed, the address pointers may ``collide''.
However, if the pointers move apart sufficiently far, the link will resume to operate as intended.
Note that the pointers colliding and moving at the same speed (i.e., the clocks running at the same speed) is an unstable equilibrium state, as the control logic aims at ``pushing'' them apart.
Accordingly, this is a metastable state of the link, which can be expected to resolve fairly quickly.

We used a variation of the Spice simulation that allows us to initialize sender
and receiver clocks to a specific offset (due to the the machinery simulation does
not start exactly at $0$\,ns). Together with a suitable initialization of the
full/empty flags, this simulates one of the clocks being started earlier.

We simulated the link with small initial offsets of the continuous pointers, i.e.,
$p_s(0)-p_r(0)=C_s(0)-C_r(0)+N/2\approx 0$, with the goal of finding a good tradeoff between
resolution time and precision of the initialization.
\figref{fig:phaseoffset} shows the pointer offset of the sender and the receiver
clock ($p_s(t)-p_r(t)-N/2$) over time $t$ for different initializations.
We see that simulations with an initial offset of $0$\,ps, $30$\,ps and $50$\,ps
stay in the metastable state until eventually the sender advances by one clock
cycle relatively to the receiver and the simulation stabilizes. Similarly, a
simulation with an initial offset of $-75$\,ps stays in the metastable state
until the receiver advances by one clock cycle relative to the sender and the
simulation reaches the corresponding stable state. Simulations with $30$\,ps
resp.\ $-75$\,ps offsets resolve after $11$\,ns resp.\ $10$\,ns. Hence, if the
designer is willing to wait $11$\,ps after initialization, it is sufficient to
guarantee avoiding this window of $105$\,ps during initialization. At the given
clock speed, this corresponds to a much larger $\delta = 1/f^+ - 105/2\,$ps,
which in our setting is roughly $360$\,ps. In general, waiting for a couple of
clock cycles after initialization increases the slack $\delta$ to being close to
a full clock cycle.

\begin{figure*}[h!]
  \includegraphics[width=\textwidth,clip,trim=6 70 3 70]{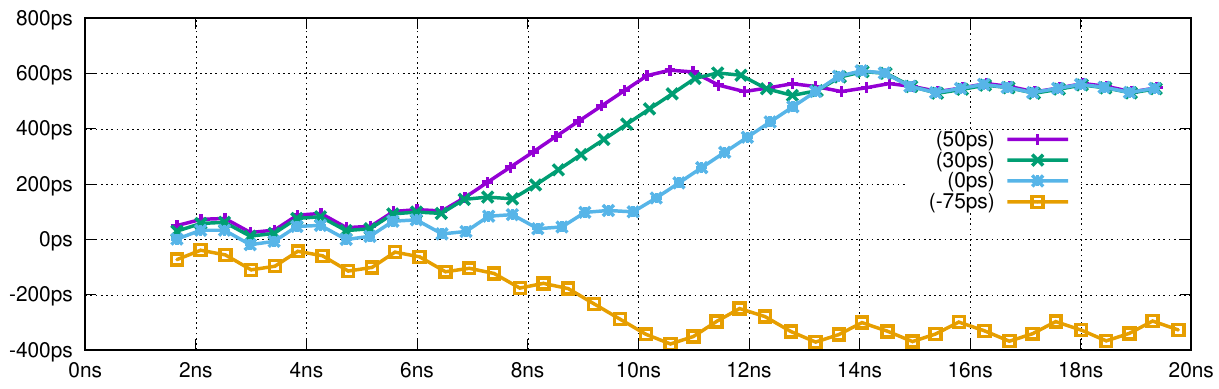}
  \vspace{-0.3cm}
  \caption{Offset of the sender and receiver pointers over simulation time. When initializing the
  pointer offset to $50$, $30$, $0$ and $-75$\,ps, we observe different times to stabilization. According
  to the analysis, setup and hold times cannot be violated once link is stabilized.}
  \vspace{-0.3cm}
  \label{fig:phaseoffset}
\end{figure*}
% END \input{simulations.tex}

\section{Conclusion}\label{sec:conclusion}
We provided a digitally controlled implementation of a synchronously accessed buffer-based link, where both sender and receiver each have their own tunable clock.
This can be seen as a distributed phase-locked loop, as we guarantee a fixed bound on the absolute time difference between the clocks, based on feedback derived from measuring the phase difference via keeping book of buffer accesses.
Our design is novel in that we neither rely on analog or asynchronous design nor incur synchronizer delay, yet deterministically guarantee correct operation.
By accepting an un- or metastable control signal for the oscillators when the buffer is roughly half full, we can completely dispense with synchronizers in the control loop.
As this eliminates the associated delay from the control loop, it leads to relaxed timing constraints compared to synchronizer-based solutions.
As a result, our link implementation can operate with a minimal buffer size of $2$ under fairly weak requirements on the frequency stability of oscillators, yet guarantee correctness deterministically.
We complemented our formal claims with VHDL and Spice simulations of UMC 65\,nm ASIC implementations.

For the link to operate correctly upon initialization, it may be the case that $\delta$, the initial clock offset between the sender and receiver clock, needs to be fairly small (cf.~Corollary~\ref{cor:clocked_correct}).
If the resulting constraint is too tight, one can violate this constraint, possibly resulting in the two address pointers ``meeting'' each other.
However, this is a metastable state of the control loop: If the two pointers move apart sufficiently far, operation will go back to the intended mode and push the pointers apart.
Note that once the link has stabilized, the resulting total clock difference between the sender and receiver clock is unknown.
One could now use the operational link to let the sender communicate its current clock value to the receiver (prefixing the encoding e.g.\ by a $1$, while the buffer cells where initialized to $0$).

However, the most practical compromise may be to avoid this complication and simply relax the initialization constraint without removing it entirely, as discussed in Section~\ref{sec:slack}. Simulating the link with varying initial pointer offsets, we demonstrated a reasonable tradeoff between the time the link stays in an unstable state (max $11$\,ns) and the precision of the initialization (in two clock cycles avoid a window of $105$\,ps), cf.~Figure~\ref{fig:phaseoffset}.

One limitation of the proposed system is that it is restricted to a single link.
% Naively designing a communication network using our proposed solution would require nodes to have a separate clock for each link, failing to transfer the data network nodes receive on different links into the same clock domain.
% In future work, we intend to address this by extending our techniques to run a clock synchronization algorithm.
 % in such a network, where each node is using a single tunable oscillator for all its incoming and outgoing links.
In follow-up work~\cite{async20pals}, the ideas presented here are combined with a gradient clock synchronization algorithm~\cite{kuhn10gradient,lenzen10tight} that tightly bounds the phase offset between adjacent nodes.
This retains the advantages of small buffers and latency while maintaining deterministic correctness.
Future work needs to flesh this concept out into a fully-fledged design, which subsequently is to be tested in silicon.
Here, suitable oscillators are more challenging to devise, because the scalability of the system is directly affected by the parameters of the oscillators.
Ultimately, the result will be an alternative approach to clocking synchronous systems with far better scalability properties than classic designs, which derive time from a single reference.

\section*{Acknowledgment}
We thank Attila Kinali and Prof.\ Ran Ginosar for pointing out that our link
controller is in fact an ``all-digital PLL.''

\bibliographystyle{IEEEtran}


\vfill
\pagebreak

\begin{IEEEbiography}
[{\includegraphics[width=1in,height=1.25in,clip,keepaspectratio,trim=5 0 4 0]{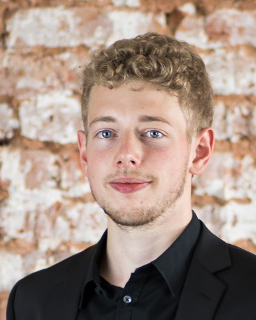}}]{Johannes~Bund}
is a Ph.\,D.\ student at the Algorithms and Complexity department at MPI for
Informatics. He graduated his B.\,Sc.\ and M.\,Sc.\ studies at the Saarland
Informatics Campus. In 2018 he joined Christoph Lenzen's group at MPI for
Informatics.
\end{IEEEbiography}

\begin{IEEEbiography}
[{\includegraphics[width=1in,height=1.25in,clip,keepaspectratio]{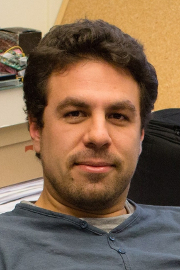}}]{Matthias~F\"ugger}
received his M.\,Sc.\ (2006) and his PhD (2010) in computer engineering from TU Wien, Austria.
He worked as an assistant professor at TU Wien, and as a post-doctoral researcher at LIX, Ecole polytechnique
and at MPI for Informatics.
Currently, he is a CNRS researcher at LSV, ENS Paris-Saclay, and
  co-leading the Digicosme group HicDiesMeus on Highly Constrained Discrete Agents for Modeling Natural Systems.
\end{IEEEbiography}

\begin{IEEEbiography}
[{\includegraphics[width=1in,height=1.25in,clip,keepaspectratio]{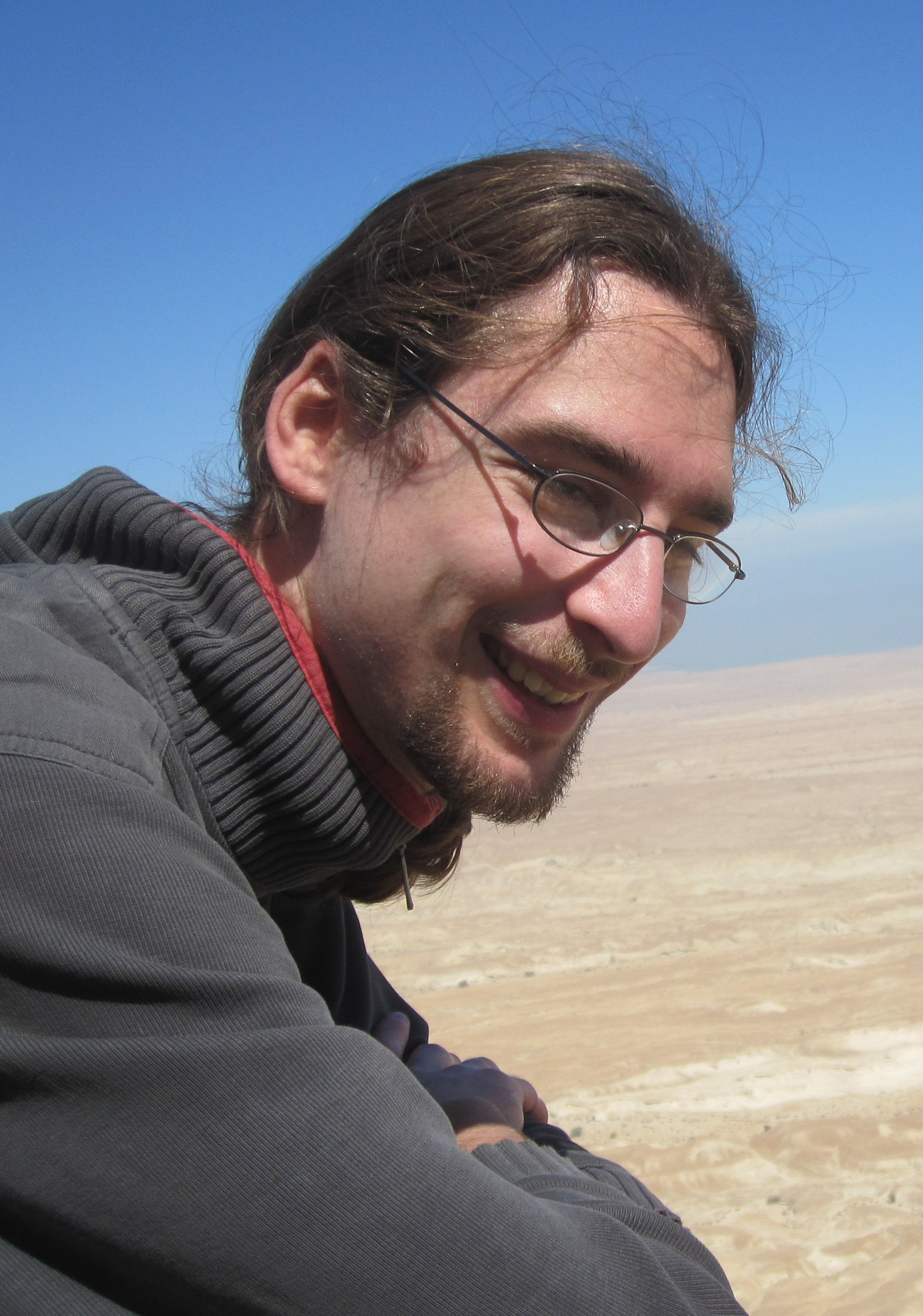}}]{Christoph~Lenzen}
received a diploma degree in mathematics from the University of Bonn in 2007 and a
Ph.\,D.\ degree from ETH Zurich in 2011. After postdoc positions at the Hebrew University of Jerusalem,
the Weizmann Institute of Science, and MIT, he became group leader at MPI for Informatics in 2014.
He received the best paper award at PODC 2009, the ETH medal for his dissertation, and in 2017 an ERC starting grant.
\end{IEEEbiography}

\begin{IEEEbiography}[{\includegraphics[width=1in,height=1.25in,clip,keepaspectratio]{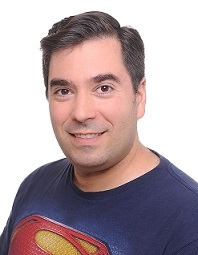}}]{Moti Medina}
is a faculty member in the School of Electrical \& Computer Engineering at
the Ben-Gurion University of the Negev since 2017. Previously, he was a post-doc
researcher in MPI for Informatics and in the Algorithms and Complexity group at
LIAFA (Paris 7). He graduated his Ph.\,D., M.\,Sc., and B.\,Sc.\ studies at the
School of Electrical Engineering at Tel-Aviv University, in  2014, 2009, and 2007
respectively. Moti is also a co-author of a  text-book on logic design
``Digital Logic Design: A Rigorous Approach'', Cambridge Univ. Press, 2012.
\end{IEEEbiography}

\vfill

\end{document}